\documentclass[a4paper]{llncs}
\usepackage{times}
\usepackage{amssymb}
\usepackage{epsfig}
\usepackage{stmaryrd}
\usepackage{float}
\usepackage{textcomp}
\usepackage{url}
\usepackage{listings}
\usepackage{graphicx}
\usepackage{wrapfig}
\usepackage{floatflt}
\usepackage{verbatim}
\usepackage{color}
\usepackage[boxruled,vlined,english]{algorithm2e} 
\usepackage{enumerate}
\usepackage{xspace}

\newcommand{\pf}{{\sc Proof} }

 \newcounter{ncomm}
\newcommand\pnote[1]{\textcolor{red}{\textbf{#1}}}
\newcommand\poldnote[1]{\emph{\textcolor{black}{\textbf{#1}}}}
\renewcommand\pnote[1]{}
\renewcommand\poldnote[1]{}

\newcommand{\oStar}{\ensuremath{o^*}\xspace}

\begin{document}
\pagestyle{plain}

\title{Mobile Agents Rendezvous in spite of a Malicious Agent~\thanks{This work has been  partially supported by the PRIN 2010 Project {\em Security Horizons}, and by the ANR - MACARON project (anr-13-js02-0002).} 
}

\author{
Shantanu Das\inst{1} 
\and
Flaminia L. Luccio\inst{2}
\and
Euripides Markou\inst{3}
}

\institute{%
LIF, Aix-Marseille University, Marseille, France 
\and
DAIS, Universit\`a Ca' Foscari Venezia, Venezia, Italy 
\and Department of Computer Science and Biomedical Informatics, \\ University of Thessaly, Lamia, Greece 
}

\maketitle

\setcounter{footnote}{0}

\begin{abstract}
We examine the problem of rendezvous, i.e., having multiple mobile agents gather in a single node of the network.
Unlike previous studies, we need to achieve rendezvous in presence of a very powerful adversary, a malicious agent that moves through the network and tries to block the \emph{honest} agents and prevents them from gathering. The malicious agent can be thought of as a \emph{mobile fault} in the network. The malicious agent is assumed to be arbitrarily fast, has full knowledge of the network and it cannot be exterminated by the honest agents.
On the other hand, the honest agents are assumed to be quite weak: They are asynchronous and anonymous, they have only finite memory, they have no prior knowledge of the network and they can communicate with the other agents only when they meet at a node. 
Can the honest agents achieve rendezvous starting from an arbitrary configuration in spite of the malicious agent? We present some necessary conditions for solving rendezvous in spite of the malicious agent in arbitrary networks. We then focus on the ring and mesh topologies and provide algorithms to solve rendezvous. For ring networks, our algorithms solve rendezvous in all feasible instances of the problem, while we show that rendezvous is impossible for an even number of agents in unoriented rings. For the oriented mesh networks, we prove that the problem can be solved when the honest agents initially form a connected configuration without holes if and only if they can see which are the occupied nodes within a two-hops distance. To the best of our knowledge, this is the first attempt to study such a powerful and mobile fault model, in the context of mobile agents. Our model lies between the more powerful but static fault model of \emph{black holes} (which can even destroy the agents), and the less powerful but mobile fault model of \emph{Byzantine agents} (which can only imitate the honest agents but can neither harm nor stop them). \\


\end{abstract}

\section{Introduction}
\label{intro}

One of the fundamental problems in distributed computing with mobile robots or agents is the problem of gathering all agents at a single location, known as the \emph{rendezvous} problem. Rendezvous is important for example, for coordination between the agents or for sharing information or for planning a collaborative task. 
This  problem has been well studied for the fault-free environment but there are very few results on solving rendezvous in the presence of faults, in particular, in the presence of a hostile entity that could prevent the agents from achieving their task.  
As in most previous works, we model the environment as a connected graph with multiple mobile agents moving along the edges of the graph; the objective is to gather them at a single node of the graph. In this context, the hostile entity may be either stationary (e.g. a harmful node in the graph) or mobile (e.g. a virus propagating on a network). Methods for protecting mobile agents from malicious host nodes have been proposed, e.g. based on the identification of the malicious host~\cite{DBLP:journals/algorithmica/DobrevFPS07}. However, the issue of protecting a network (hosts and mobile agents) from a malicious and mobile entity is still wide open (see \cite{:FlocchiniS2012} and references therein).

A model for a particularly harmful node
which has been extensively studied is the \emph{black hole},  where a node which contains a stationary process destroys all mobile agents upon visiting it, without leaving any trace. In this case, although the hostile entity is very powerful, it is stationary; the mobile agents can simply avoid the black hole once its location is known. Thus, the main issue is locating the black hole \cite{DBLP:journals/algorithmica/DobrevFPS07,DBLP:journals/algorithmica/FlocchiniIS12,kmrs07}. Locating and avoiding a malicious entity that is also mobile and moves from node to node of the graph, seems to be a more difficult problem.
A recent result considers the problem of rendezvous in the presence of \emph{Byzantine agents}~\cite{dpp14}. A Byzantine agent is indistinguishable from the legitimate or \emph{honest} agents, except that it may behave in an arbitrary manner and may provide false information to the honest agents in order to induce them to make mistakes, thus preventing the rendezvous of the honest agents. Thus, the issue here is identifying the Byzantine agents and distinguishing them from the honest agents. Note that the Byzantine agent cannot actively harm the agent or physically  prevent the agents from gathering. In this paper, we consider a more powerful adversary called a \emph{malicious agent} which can actively block the movement of an honest agent to the node occupied by the malicious agent. For example, when two honest agents are close to each other, the malicious agent can enter the path between the two agents and prevent them from meeting. We investigate the feasibility of rendezvous in the presence of such a powerful adversary. In particular, the malicious agent is more powerful than the honest agents; it can move arbitrarily fast through the graph, has full information about the current configuration (i.e. the graph and location of the agents), and has knowledge of the next action to be taken by each honest agent. On the other hand, the honest agents are relatively weak; they are anonymous finite automata, they move asynchronously without any prior knowledge of the graph and they can communicate only locally on meeting another agent at the same node.
We remark here that the malicious agent is distinguishable from the honest agents, so the question of identifying the malicious agent (as in Dieudonne et al.~\cite{dpp14}), does not arise here.

We believe this is an interesting model for studying mobile faults in a graph, that has never been considered before. In some sense this model can be seen as an extension of the model of networks with delay faults. For example, Chalopin et al.~\cite{cdlp14} consider the problem of rendezvous in the presence of an adversary that can prevent an agent from moving for an arbitrary but finite time. In their case, the agent cannot be blocked forever as in our scenario.
Our model can also be contrasted with the model for network decontamination or, cops and robbers search games on graphs, where a team of good agents (called cops) tries to capture a fast fugitive (robber). The fugitive or hostile entity is exterminated as soon as one of the cops reaches it. Thus the behavior of the hostile entity, in this case, is opposite to that of the malicious agent in our model -- instead of blocking the honest agents, the hostile entity tries to get away from the good agents.     

In terms of practical motivation for this research, we can think of the malicious agent as representing a virus that may spread around the network. While in the classical decontamination problem the aim is to extinguish  the virus, in our setting the virus cannot be extinguished and has to be contained in one part of the network, thus dividing the network into {\em unstrusted} and {\em trusted} subnetworks. This scenario can be compared to the problem of \emph{botnets}, i.e. a subnet of compromised computers (bots), typically used for denial-of-service attacks on the internet.    
The untrusted subnetwork in our model can be seen as a botnet, and the \emph{botmaster} who controls the bots represents the malicious agent.
An honest agent that resides on a node protects the trusted network from the untrusted one by running some protection mechanism (e.g. a firewall, an intrusion detection mechanism, etc.). Thus the malicious agent cannot enter a node already occupied by an honest agent.
On the other hand the botnet is dynamic, and it may reduce its dimension (i.e., when the botmaster leaves the host) or it may increase it only on hosts not occupied, i.e., not protected by an agent.  Honest agents may expand towards the untrusted hosts which are not controlled by the botmaster anymore by running botnet detection mechanisms (see, e.g., \cite{YXS10}). 
We are then interested in solving the rendezvous problem in the  trusted subnetwork, and we want to study how this malicious behaviour affects the solvability of the {\em Rendezvous} problem.

\paragraph{\bf Related Work: }

The rendezvous problem has been studied for agents moving on graphs~\cite{ag02} or for robots moving on the plane~\cite{cp08}, using either deterministic or randomized algorithms.
In the fault-free scenario, the rendezvous problem can be solved relatively easily, even in asynchronous networks, when the network has an  asymmetry (e.g., a distinguished node), and can be explored by the agents, since the mobile agents can simply be instructed to meet at such a distinguished node. However, this is not the case for symmetric networks, or when the agents is incapable of visiting all nodes of the network, and the rendezvous problem in such settings is non-trivial and not always solvable even in simple topologies such as the ring network~\cite{kkmbook}. Symmetry-breaking for the rendezvous problem can be achieved by attaching unique identifiers to the agents (see, e.g.,~\cite{clp10,yy96}), or in the anonymous case using tokens as in e.g., \cite{cd10,dmsvw08}.
With respect to hostile environments, the {\em Rendezvous} problem has been studied when there is a black hole or other stationary faults in the network~\cite{cds07,dfps03,YIK12}. 
Another model for hostile nodes 
has been presented in~\cite{blmpp14,DBLP:conf/sirocco/KralovicM10}, where the authors have studied how a more severe (than a black hole) behaviour of a malicious host affects the solvability of the \emph{Periodic Data Retrieval} problem in asynchronous networks. A well studied problem in the context of a mobile adversary is the problem of graph searching where a team of mobile agents has to decontaminate the infected sites and prevent any reinfection of cleaned areas. This problem is equivalent to the one of capturing a fast and invisible fugitive moving in the network. 
For results on this and related problems see, e.g., \cite{bfffnst12,fhl07,fhl08,l09}.

Gathering of mobile agents has been also studied in the plane when there are faulty agents which may crash~\cite{ap06,Bouzid0T13} and in networks with delay faults~\cite{cdlp14} or in the presence of Byzantine agents~\cite{dpp14}, as mentioned before. 
However, to the best of our knowledge, the rendezvous problem has never been studied under the presence of hostile agents that may block other agents from having access to parts of the network.

\paragraph{\bf Our Results: }
In this paper we consider a network modelled as a connected undirected graph with multiple honest agents located at distinct nodes of the graph.
There is also a hostile entity which is mobile, called the malicious agent. It cannot harm the honest agents but can prevent them from visiting a node: an honest mobile agent cannot visit a node which is occupied by a malicious agent and vice versa. 
We are interested in solving the rendezvous of all honest agents in this hostile environment. Our objective is to study the feasibility of rendezvous with minimal assumptions. Thus, we consider the weakest possible model for the honest agents. The honest agents are finite state automata with local communication capability and having no prior knowledge of the network. In Section~\ref{impossibility} we show some configurations in which the problem is unsolvable  and we discuss properties that must be respected by any correct algorithm for the problem. For the rest of the paper, we consider ring and mesh networks -- two topologies that can be explored even by a finite automaton.
In Section~\ref{ring} we present a rendezvous algorithm for ring networks. For oriented rings, we have a universal algorithm that achieves rendezvous starting from any initial configuration, despite the existence of a malicious agent. We prove that the problem is unsolvable for any even number of agents in unoriented rings. Finally, we present an algorithm for rendezvous of  any odd number of agents in unoriented rings, thus solving the problem in all solvable instances. 
In Section~\ref{mesh-tori} we consider oriented mesh topologies and we prove that the problem can be solved when the agents initially form a connected configuration without holes if and only if they can detect which are the occupied nodes within a distance of two hops. We show that this latter capability is necessary to achieve rendezvous even for connected configurations without holes.
We conclude in Section~\ref{conclusion} with a discussion about future research directions for this new model.

\section{Preliminaries}
\label{impossibility}

\subsection{Our Model}
\label{model}

We represent the network by a graph $G=(V,E)$ composed by $|V|=n$ anonymous nodes or {\em hosts} and $|E|$ edges or connections between nodes. Each host is connected to other hosts by bidirectional asynchronous FIFO
links (i.e., an agent cannot overtake another agent moving in the same edge), and it is capable of serving agents by a mutual exclusive mechanism (i.e., an agent at a node $u$ must finish its computation and move or decide to stay, before any other agent at $u$ starts its computation or another agent visits $u$). The links incident to a host are distinctly labelled but this port labelling (unless explicitly mentioned), is not globally consistent. In the network there are some {\em mobile agents} which are independent computational processes with some constant internal memory. The agents  are initially scattered in the network (i.e., at most one agent at a node), 
and can move along its edges. 
An agent arriving at a node $u$, learns the label of the incoming port, the degree of $u$ and the labels of the outgoing ports. 
We assume there are $k$ {\em honest} anonymous identical agents $A_1, A_2, \ldots A_{k}$, and one {\em malicious} agent $M$ which may deviate from the proper operations. The initial locations of the honest and malicious agents are decided by an adversary.
We describe below the capabilities and behaviour of honest and malicious agents.

\medskip

\noindent
{\bf Honest agents:} An honest agent located at a node $u$ can see all other agents at $u$ (if any), and can also read their states. It can also read the degree of $u$ and the labels of the outgoing ports. The agents are anonymous, cannot exchange messages and cannot leave messages at nodes. They are identical finite state automata, hence they have some constant memory. 
The agents do not know $n$ and $k$. 
Two agents travelling on the same edge in different directions do not notice each other, and cannot meet on the edge. Their goal is to rendezvous at a node.
 
 \medskip
 
 \noindent
{\bf Malicious agent:} We consider a worst case scenario in which the malicious agent $M$ is a very
powerful entity compared to honest agents: It can move arbitrarily fast inside the network (since the model is asynchronous and the adversary is combined with the malicious agent) and it can permanently
`see' the positions of all the other agents. It has unlimited memory and knows the transition function of the honest agents. When it resides at a node $u$ it prevents any honest agent $A$ from visiting $u$ (i.e., it ``blocks" $A$): if an agent $A$ attempts to visit $u$ it receives a signal that $M$ is in $u$ (botnet detection) and in that case we say that $A$ {\em bumps} into $M$. The malicious agent can neither visit a node which is already occupied by some honest agent, nor cross some honest agent in a link. It also obeys the FIFO property of the links (i.e., it cannot overpass an honest agent which is moving on a link).

 \medskip
 
We call a node $u$ {\em occupied} (respectively,  {\em free} or {\em unoccupied}) when one or more (no) {\em honest} agents are in $u$.
We notice here that some of our impossibility results hold even for stronger models, e.g.,  when honest agents have unlimited memory, distinct identities, knowledge about the size of the network, visibility, etc. Our algorithm for the ring topology only requires the capabilities of the honest agents mentioned above while for the mesh topology we assume that the honest agents also have the ability to scan whether a node within a two-hops distance, is occupied or not.

\subsection{Basic Properties} 
In this section we show a special class of configurations for which the problem is unsolvable. Intuitively, those are configurations in which the malicious agent can keep separated at least two agents forever.

\begin{definition}\label{separable}
Let ${\cal C}$ be a configuration of a number of agents in a graph $G$ with a malicious agent. The configuration ${\cal C}$ is called {\bf separable} if there is a connected vertex cut-set $F$ composed of free nodes which, when removed, disconnects $G$ so that not all occupied nodes are in the same connected component.
\end{definition}

\begin{lemma}\label{impossibility-initial}
Rendezvous is impossible for any initial configuration in a graph $G$ which is separable, even if the agents have unlimited memory, distinct identities and can always see their current configuration.
\end{lemma}
\begin{proof}
Let ${\cal C}$ be an initial configuration which is separable, and let $F$ be a connected vertex cut-set, whose removal disconnects $G$ so that not all occupied nodes are in the same connected component. Let $u, v$ be two occupied nodes which are in different connected components of $G$ and let $A, B$ be the honest agents located at $u,v$ respectively. Due to asynchronicity an adversary can introduce delays to $A$'s and $B$'s movements while at the same time the malicious agent, which has been initially placed at a node in $F$, can move everywhere in $F$ (since $F$ has only free nodes and it is connected) preventing agents $A, B$, from visiting any node in $F$. Since all paths between $u$ and $v$ include at least one node of $F$, agents $A, B$ can never meet, no matter how powerful they are.\qed
\end{proof}

Hence for every initial separable configuration the problem is unsolvable. A natural question is whether there are non-separable initial configurations for which the problem is unsolvable. The answer is yes and one can easily find such configurations (see Examples \ref{connected-sep} and \ref{disconnected-sep} below).

We now state sufficient conditions under which the problem is unsolvable for a separable (initial or not) configuration of agents.

\begin{definition}\label{separating}
Let ${\cal C}_t$ be a configuration at time $t \geq 0$ (i.e., initial or not) of a number of agents in a graph $G$ with a malicious agent. The configuration ${\cal C}_t$ is called {\bf separating} if ${\cal C}_t$ is separable and either ${\cal C}_t$ is an initial configuration or the following conditions hold:
\begin{itemize}
\item there is a node $x_m \in F_t$ ($F_t$ is any vertex cut-set of ${\cal C}_t$ as defined in Definition~\ref{separable}) and a path of nodes $(x_0, x_1, \ldots , x_m)$ so that $x_0$ is free at time $0$ and, 
\item the sequence of nodes $(x_0, x_1, \ldots , x_m)$ can be partitioned in $k \leq t+1$ contiguous subsequences $(x^0_0, \ldots , x^0_i), (x^1_{i+1}, \ldots , x^1_j), \ldots ,(x^k_{l+1}, \ldots , x^k_m)$, where $0 \leq i < j < l < m$ and, 
\item the nodes $(x^s_w, \ldots , x^s_r)$ belonging to the same subsequence $s$ are free at time $s$, where $0 \leq s \leq k$ and nodes $(x^k_w, \ldots , x^k_r)$ are free at time $t$.
\end{itemize}
\end{definition}
\begin{lemma}\label{impossibility-anytime}
Rendezvous is impossible for any {\em separating} configuration in a graph $G$, even if the agents have unlimited memory, distinct identities and can always see their current configuration.
\end{lemma} 
\pf
Let ${\cal C}_t$ be a separating configuration of agents in a graph $G$ with a malicious agent. If ${\cal C}_t$ is an initial configuration then by Lemma~\ref{impossibility-initial}, rendezvous is impossible. Assume that ${\cal C}_t$ is a configuration that has resulted after some moves of the honest agents. Since ${\cal C}_t$ is a separating configuration, the malicious agent can be initially placed at $x_0$  and at time $1$ can move from the last node of subsequence $0$ to the last node of subsequence $1$ via the path of free nodes at time $1$. Then at time $2$ it can move from the last node of subsequence $1$ to the last node of subsequence $2$, and so on up to time $s$ when it can move to node $x_m$ where remains until time $t$. 
Since at time $t$, the configuration ${\cal C}_t$ is separable and $x_m \in F_t$, for the same arguments  presented in the proof of Lemma \ref{impossibility-initial}, there are at least two agents which will never meet.\qed

Intuitively, Lemma~\ref{impossibility-anytime} states that if ${\cal C}_t$ is a separable configuration, and in ${\cal C}_t$ there is a free node $x$ so that either: i) $x$ has been always free or, ii) there are paths of nodes which eventually become free and they form a connection between a free node at ${\cal C}_0$ and $x$, then there are at least two agents in ${\cal C}_t$ which will never meet.
Hence, any correct algorithm for the solution of the problem should avoid creating a separating configuration.

\begin{example}\label{connected-sep}
Suppose that the set of occupied nodes is connected and the configuration is as shown in Fig.~\ref{cut} left (the occupied nodes are denoted by circled nodes). 
\begin{figure}
\vbox to 6mm{\hspace{0cm} \psfig{figure=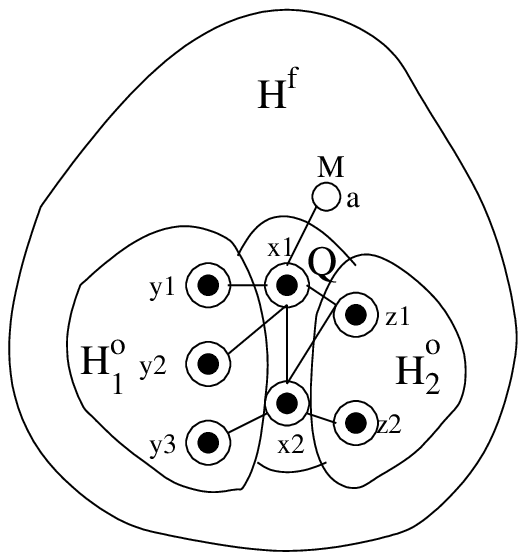,width=30mm}
\hfill \hbox to 50mm{\hspace{0cm}
\psfig{figure=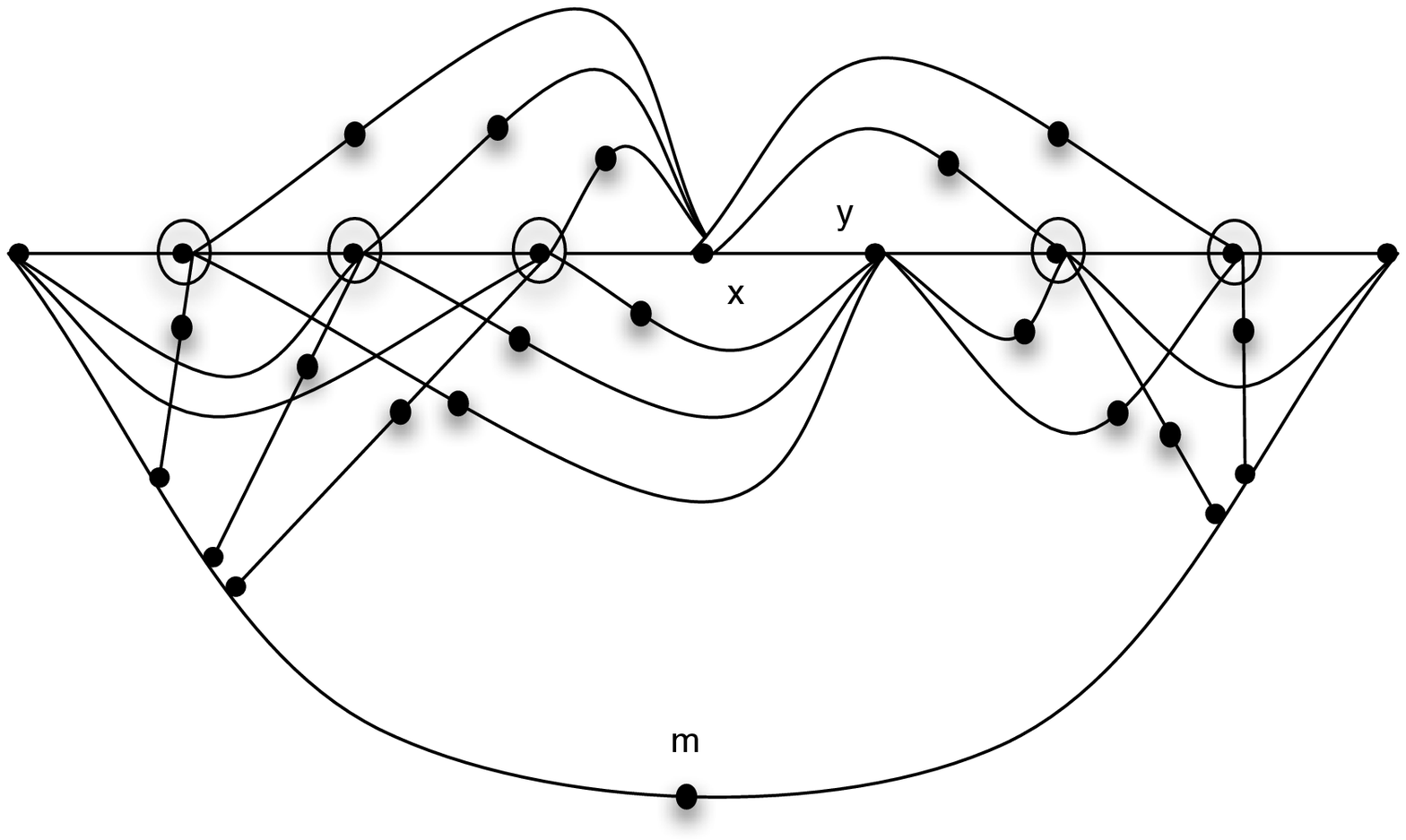,width=52mm}}\vss}
\vspace*{3.3cm}
\caption{\label{cut} Left: $H^f$ and  $H^o$ are the subgraphs induced by all free and occupied nodes, respectively.  Let $Q$ be a vertex cut-set of $H^o$ consisting of the nodes $x_1,x_2$. The removal of nodes in $Q$ splits $H^o$ into $H^o_1$ and $H^o_2$. The malicious agent $M$ resides at $a$ and there is an edge $\{x_1,a\}\in E(G)$. Right: Initially, nodes $x,y$ are unreachable by the malicious agent $M$ located at node $m$. If any agent moves, then $M$ can reach both nodes $x,y$ and all new free nodes and therefore prevent the agents on the left to meet the agents on the right.}
\end{figure}
Any correct algorithm for the rendezvous problem should avoid moving the agents from both $x_1,x_2$ nodes as long as there are agents in $H^o_1$ and $H^o_2$ since in that case the resulting configuration will be separating. In other words, in such a case the malicious agent will prevent forever the meeting of any agent in $H^o_1$ with an agent in  $H^o_2$. However, due to the asynchronicity of the agents' movements and the limited knowledge the agents have about the configuration, an algorithm cannot guarantee such a behaviour (even if the agents have unlimited memory and distinct identities).\qed
\end{example}

\begin{example}\label{disconnected-sep}
Consider now the initial configuration ${\cal C}$ in graph $G$ of Fig.~\ref{cut} right. There are sets $F_i$ of free nodes whose removal would disconnect $G$ so that not all occupied nodes are in the same connected component. However every such set $F_i$ is disconnected. Hence the configuration is not separable.
The agents are in the circled nodes of Fig.~\ref{cut} right, and the malicious agent is in node $m$. If any agent moves, then the new subgraph of occupied nodes remains disconnected, but the new subgraph of free nodes (whose removal disconnects $G$ so that not all occupied nodes are in the same connected component), becomes connected and thus the resulting configuration is separating. In other words, if any agent moves then the malicious agent will prevent forever the meeting of any two agents occupying nodes left and right of node $x$.
Hence no matter how much power or capabilities the agents have (even if they can see the configuration at all times), rendezvous is impossible.\qed
\end{example}

\section{Rendezvous in a Ring Network}
\label{ring}

In this section we will study the rendezvous problem in bidirectional rings with a malicious agent $M$.
Notice that in a ring topology there are no separable (and hence no separating either) configurations, since there cannot exist a connected cut-set composed of free nodes whose removal would disconnect the ring. However, 
since the ring is highly symmetric, rendezvous is impossible even if the agents have unlimited memory and have full knowledge of the configuration, since an adversary can keep synchronized the agents so that they always take the same actions at the same time and therefore they maintain their initial distances (the malicious agent can keep on moving synchronized with the honest agents). 
Thus, in order to solve the problem we need to add some constraints to the model. A natural step is to assume that there is a special node labeled \oStar in the ring which can be recognized by the agents. Note that the malicious agent is so powerful that it could place itself on \oStar and never move from there. Our strategies also work under this scenario. 
We now present algorithms for solving the problem in oriented and unoriented rings.

\subsection{Oriented Ring}

In an oriented ring, the two incident edges at each node are labelled as clockwise or counter-clockwise in a consistent manner; so, all agents agree on the ring orientation. 

The idea of the algorithm is the following. Each agent moves in the clockwise direction until it meets \oStar or bumps into $M$. For the first three times that the agent bumps into $M$ without meeting \oStar, it reverses its direction and continues moving in the opposite direction. Due to the FIFO property and the fact that the agent cannot pass over $M$, we can show that if an agent bump into $M$ after reversing directions at least three times, then the other agents should have bumped into $M$ at least twice, without meeting the special node \oStar (see Lemma~\ref{lemma:2bumps}). After an agent has already bumped into $M$ three times, the next time it bumps into $M$ or meets \oStar it stops. On the other hand, if the agent meets \oStar before it bumps into $M$ twice, then the agent stops at \oStar, and all the other agents will arrive at \oStar after bumping into $M$ at most once. The algorithm called RV-OR  is presented below. 
 
\begin{algorithm}

Let $i := 0$\;
DIR := Clockwise\;
\While{not Stopped}{

  Move DIR until you bump into $M$ or meet \oStar or a stopped agent\;  i=i+1\;
  \If {you met a Stopped agent} {Become Stopped and Exit loop\; }
  \If {$i=1$ or $i=2$}
  {
    \If {Current node is \oStar } {Become Stopped and Exit loop\;}
    \ElseIf{Bumped into $M$}  { Reverse direction (DIR := inverse(DIR))\; }
  }
  \If {$i=3$}
  { 
     \If {Current node is \oStar or Bumped into $M$} { Reverse direction (DIR := inverse(DIR))\; }
  }
  \If {$i=4$}
  { 
     \If {Current node is \oStar or Bumped into $M$} { Become Stopped and Exit Loop \; }
  }
}  

 \caption{RV-OR : Rendezvous of $k\geq2$ agents in oriented rings}
  \label{algo:ring_async}
\end{algorithm}

\begin{lemma}\label{lemma:2bumps}
During the execution of the algorithm, if an agent bumps into $M$ in the fourth iteration of the {\em while} loop, then any other agent must have bumped into $M$ at least two times. 

\end{lemma}
\pf 
Consider the first agent $A$ that bumps into $M$ two times at nodes $u,v$ without meeting \oStar and consider the segment of the ring between these two nodes $u$ and $v$, where $A$ was moving until it bumped into $M$ for the second time. The special node  \oStar  cannot be in this segment, otherwise the agent $A$ would have stopped there. If another agent $B$ meets \oStar without having bumped into $M$ (i.e., going clockwise), then it will stop there, and when agent $A$ or any other agent changes direction at the second bump, it will reach the Stopped agent $B$, before bumping into $M$ for a third time. So no agent reaches the fourth iteration in this case. Now suppose that an agent $C$ meets \oStar after having bumped into $M$ once (i.e. going counter-clockwise), then agent $C$ stops at \oStar. At this time, agent $A$ must be in its third iteration (having bumped twice into $M$). Agent $A$ will bump into $M$ for the third time and reverse direction. Now agent $A$ will reach \oStar and the Stopped agent $C$ before reaching $M$ in its fourth iteration. So, if the agent $A$ bumps into $M$ in its fourth iteration, this implies that no agent meets \oStar in the first two iterations i.e. all other agents have bumped into $M$ twice.
\qed

\begin{lemma}\label{lemma:oriented}
Algorithm RV-OR solves rendezvous of $k\geq 2$ agents in spite of one malicious agent, in any oriented ring containing one special node \oStar.
\end{lemma}
\pf

\noindent {\em Case 1}: No agent reaches the third iteration of the while loop. Thus no agent has bumped into $M$ more than once. In this case all agents will eventually reach \oStar and stop there. So, rendezvous is achieved.

\noindent {\em Case 2}: No agent reaches the fourth iteration. Consider the first agent $A$ to reach the third iteration of the while loop. Then this agent has bumped into $M$ twice (and it has changed direction twice). If this agent meets a Stopped agent then the Stopped agent must be in the special node \oStar. In this case any subsequent agent would also reach \oStar by the same argument and thus we have rendezvous. Thus, either rendezvous is achieved or the agent $A$ in the third iteration does not reach a Stopped agent and enters the fourth iteration of the loop. 

\noindent {\em Case 3:} If all agents that reach the fourth iteration, stop at \oStar or at another stopped agent, then we have rendezvous at  \oStar (since any agent that stops before the fourth iteration should have stopped at \oStar). So the only remaining case is when an agent  stops on bumping into $M$ in the fourth iteration. 

In this case, due to Lemma~\ref{lemma:2bumps} all other agents would have bumped into $M$ at least $2$ times, i.e they would all reach the third iteration of the loop. Since no agent in the third iteration can stop at $M$ or \oStar they either eventually reach their fourth iteration or they meet an agent which has stopped in its fourth iteration. Let $A$ be the first agent that stops in its fourth iteration.

\noindent {\em Subcase $3.1$:}We first consider the case that $A$ stops at \oStar. Let $B$ be an agent that bumps into $M$ and stops at its fourth iteration (since we are in case $3$).
Note that agents $A$ and $B$ have changed direction exactly $3$ times before reaching the 4th iteration of the loop, so the agents are moving in the same direction in the fourth iteration (even though they may not reach the fourth iteration at the same time). Let without loss of generality, that direction be {\em counter-clockwise}. Since agent $B$ stops at $M$ in its fourth iteration it could not have bumped into $M$ at the end of its third iteration because in that case it would mean that the malicious agent has disconnected the area containing the honest agents which is impossible. Hence, agent $B$ must have met \oStar at the end of its third iteration, before agent $A$ stops there (otherwise $B$ would stop at \oStar). Consider the part of the ring that starts with agent $A$ and ends with agent $B$ in counter-clockwise direction, when $B$ ends its third iteration at \oStar. If the malicious agent lies in that part then $B$ cannot bump into $M$ in its fourth iteration. If $M$ does not lie in that part then $A$ should have met \oStar at the end of its second iteration and hence $A$ should have stopped then.

\noindent {\em Subcase $3.2$:}We now consider the case that $A$ stops at $M$.
By the time $A$ stops at $M$, agent $B$ is at least in the third iteration. When this agent meets either $M$ or $\oStar$ in the third iteration, it changes its direction and now there is neither $M$ nor \oStar between this agent and agent $A$. So agent $B$ will meet agent $A$ before $B$ reaches $M$ or \oStar , thus, we will have rendezvous. 
\qed

\subsection{Unoriented Rings}
In unoriented rings, each agent has its own notion of clockwise and the agents may not agree on the clockwise direction. In this case rendezvous is not always feasible.

\begin{lemma}\label{lemma:undirected-even}
For any even number $k\geq2$, the rendezvous problem for $k$ honest agents and one malicious agent cannot be solved in any bidirectional unoriented anonymous ring with a special node \oStar, even if the agents know $k$.
\end{lemma}

\pf
\begin{figure}
 \centering \epsfig{file=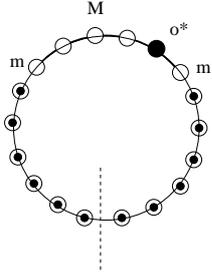, width=1.1in}
 \caption{An even number of agents are initially placed  in the lower part of the ring between $m$ and $m'$, and the  malicious agent can move in the remaining upper part . The agents cannot rendezvous. \label{ring2}}
 \end{figure}

Consider the following initial configuration fixed by an adversary: the honest agents are  initially placed in $k$ consecutive nodes and the remaining nodes, including node \oStar, are unoccupied (see Fig.~\ref{ring2}). Let $m$ and $m'$ be the two extremities of the unoccupied segment of the ring; The malicious agent $M$ is initially placed somewhere on this segment and thus, it can move between nodes $m,m'$ preventing any honest agent from visiting those nodes. The adversary splits the chain of consecutive initially occupied nodes into two subchains of equal length and forces any two agents to have a different orientation of the ring if and only if they belong to different subchains. The configuration is symmetric with an axis of symmetry crossing the chain of agents in the middle. Each honest agent has a symmetric counterpart on the other side of the axis. 
Suppose there is an algorithm $\cal A$ which solves rendezvous.  
Irrespective of the actions of the algorithm, the adversary can synchronize the agents so that at any given time, any two symmetric agents which initially belong to different subchains have the same input: They always arrive at symmetric nodes through the same (perceived) direction, and symmetric nodes are either unoccupied, occupied by the same number of agents, or adjacent to $m$ and $m'$ (i.e., the agents bump into $M$). Hence the configuration remains symmetric at every time step and any two symmetric agents belonging to different subchains will never occupy the same node at the same time.\qed


We now present an algorithm for solving rendezvous of $k$ agents, for any {\em odd} integer $k$, in an unoriented asynchronous ring network. Notice that in an unoriented ring, if we follow an algorithm similar to Algorithm RV-OR it is possible that the agents form two distinct groups that gather at two distinct nodes. However,  since the total number of agents is odd, exactly one of the two groups would have even number of agents, thus one of the agents of this group could move to collect all the other agents. The algorithm must ensure that there are at most two groups of agents, i.e. there are at most two distinct nodes where the agents stop in the initial phase. In our algorithm, an agent stops at \oStar only if it has seen it at least three times, while moving in the same direction. This implies that this agent has traversed the complete ring two times and while $M$ has moved at least once around the ring. So, there could be no agents moving in the opposite direction. On the other hand if some agent stops while bumping into $M$, then any agent moving in the same direction would reach this node with the stopped agent before reaching $M$ or \oStar. In all cases, there will be at most two nodes where the agents stop. 
When two or more agents have gathered at a node $v$, one of the agents called the \emph{searcher}\footnote{We select as searcher the second agent that arrives at node $v$.} reverses direction and moves to search for the other agents. The searcher only stops when it reaches the other node $w$ containing stopped agents. If the number of agents gathered at node $w$ is even then the searcher becomes a \emph{Collector} and it collects all agents and returns to node $v$. Note that the agent does not need to count the number of other agents as the algorithm depends only on the parity of the size of the group of agents. The complete algorithm, called {\sc RV-UR} is presented in a following table.

\begin{algorithm}
\begin{center}
\begin{description}
\item {\bf  Case $0$}. {\em Initial} state\\
Move clockwise until:
\begin{description}
	\item{\bf Case 0.1. You meet node \oStar unoccupied for the third time:} \\ Change state to {\em stopper};
	\item{\bf Case 0.2. You bump into $M$ trying to move from a node that hosts only you:} \\ Change state to {\em stopper};
	\item{\bf Case 0.3. You meet an agent not at node \oStar:} 
	\begin{description}
		\item{\bf Case 0.3.1. The agent you meet is alone and is a {\em stopper}:} \\ Change state to {\em transformer-1};	
		\item{\bf Case 0.3.2. Every other agent at the node is at state {\em final}:} \\ Change state to {\em stopper};
		\item{\bf Case 0.3.3. You meet a {\em stopper} and at least one agent at state {\em final}:} \\ Change state to {\em transformer-2};	
	\end{description}
\end{description}
\end{description}

\begin{description}
\item {\bf Case $1$}. State {\em transformer-1} \\
 Wait until all other agents change to state {\em final};\\
 Change state to {\em searcher};
\end{description}

\begin{description}
\item {\bf Case $2$}. State {\em searcher} \\
 Move counter-clockwise until you bump into $M$ while you try to move from a node $u$:
\begin{description}
	\item{\bf Case 2.1. You see one or more agents at $u$ and all of them are at state {\em final}:} \\ Change state to {\em stopper};	
	\item{\bf Case 2.2. You see no agent at $u$ or an agent not at state {\em final}:} \\ Change state to {\em collector};
\end{description}
\end{description}

\begin{description}
\item {\bf  Case $3$}. State {\em stopper}\\
	Wait until:
	\begin{description}
		\item{\bf Case 3.1. You see a {\em transformer-1} or {\em transformer-2}:} Change state to {\em final};	
		\item{\bf Case 3.2. You see a {\em collector}:} Follow {\em collector};	
		\item{\bf Case 3.3. You see a {\em terminator}:} Change state to {\em terminator};	
	\end{description}
\end{description}
	
\begin{description}
\item {\bf  Case $4$}. State {\em collector}\\
	Wait until every other agent at the node changes its state to {\em stopper};\\
	Collect everyone;\\
	Move clockwise collecting every agent you meet, until you meet an agent at state {\em final};\\
	Change state to {\em terminator};
\end{description}
	
\begin{description}
\item {\bf  Case $5$}. State {\em final}\\
	Wait until:
	\begin{description}
		\item{\bf Case 5.1. You see a {\em collector}:} Change state to {\em stopper};	
		\item{\bf Case 5.2. You see a {\em terminator}:} Change state to {\em terminator};	
	\end{description}
\end{description}

\begin{description}
\item {\bf  Case $6$}. State {\em transformer-2}\\
	Wait until every other agent at the node changes its state to {\em final};\\
	Change state to {\em final};
\end{description}

\begin{description}	
\item {\bf  Case $7$}. State {\em terminator}\\
	Wait until every other agent at the node changes its state to {\em terminator};\\
	Exit;
\end{description}

\end{center}
 \caption{RV-UR : Rendezvous in unoriented rings}
  \label{RV-UR}
\end{algorithm}

\begin{lemma}\label{lemma:RV3}
Consider an anonymous ring consisting of $n$ nodes, including a special node \oStar and one malicious agent. If $k \geq 2$ honest agents execute Algorithm {\sc RV-UR},  then, after a total number of  $O(kn)$ edge traversals the honest agents correctly rendezvous, if $k$ is odd.
\end{lemma}

\pf
\begin{figure}
 \centering \epsfig{file=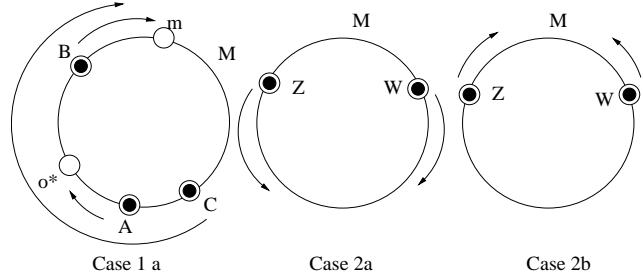, width=8.3cm}
 \caption{The cases of Lemma~\ref{lemma:RV3}. \label{f1}}
 \end{figure}
\noindent {\em Case 1}: All agents have the same orientation. Let $A$ be the first agent which exits the {\em initial} state (either by meeting node $o^*$ for the third time, or by bumping into $M$).
\begin{itemize}
\item Suppose that $A$ meets node $o^*$ (for the third time) and therefore $A$ becomes a {\em stopper} (see Fig.~\ref{f1}.1a). 
Consider the next agent $B$ which exits the {\em initial} state. Agent $B$ should bump into $M$. Let $m$ be the node where $B$ is located when it tries to move and bumps into $M$. 
$B$ becomes a {\em stopper}.
Then let $C$ be the next agent that exits the {\em initial} state. Agent $C$ should meet agent $B$ (notice that if $C$ passes from the special node $o^*$ will not stop to meet agent $A$). Thus $C$ passes from cases $(0.3.1) \rightarrow$ $(1) \rightarrow$ $(2.2) \rightarrow$ $(4) \rightarrow$ $(7)$ of the Algorithm {\sc RV-Ring}, collecting agent $A$ and ending at node $m$. Meanwhile, all other agents (if any) gather at node $m$, since they pass from cases $(0.3.2)$ and then possibly $(3.1)$, or $(0.3.3)$ and $(6)$, all of them ending at state {\em final} except possibly one which is at state {\em stopper}. Hence, when agent $C$ arrives at $m$, all the agents change to state {\em terminator} and exit.
\item Suppose that $A$ bumps into $M$ when it tries to move from a node $m$. Consider the next agent $B$ which exits the {\em initial} state. If $B$ meets node $o^*$ then we can argue similarly as in the previous paragraph by replacing agent $A$ there with $B$. If $B$ meets $A$, then we can again argue as in the previous paragraph by replacing agents $B$ and $C$ there with $A$ and $B$ respectively.
\end{itemize}

\noindent {\em Case 2}: Not all agents have the same orientation. 
We first observe that if there are at least two agents $X$ and $Y$ having a different orientation for the ring, then no agent will exit the {\em initial} state by stopping at the special node $o^*$. Indeed, let ${\cal C}^X$ and ${\cal C}^Y$ be the subsets of agents having the same orientation for the ring as agents $X$ and $Y$, respectively. Due to the presence of the malicious agent $M$ in the ring (and the fact that $M$ can not cross an honest agent), an agent from ${\cal C}^X$ (or ${\cal C}^Y$) can pass at most twice from the special node $o^*$ 
and then it has to bump into $M$:  
a) Suppose that $M$ is after an agent $Z \in {\cal C}^X$ and before an agent $W \in {\cal C}^Y$ on the direction that $Z$ is moving (see Fig.~\ref{f1}.2b).
	Then $Z$ may pass at most once from the special node $o^*$ before it bumps into $M$.  
b) Suppose that $M$ is before an agent $Z \in {\cal C}^X$ and after an agent $W \in {\cal C}^Y$ on the direction that $Z$ is moving (see Fig.~\ref{f1}.2a).
	 Then $Z$ may pass at most twice from the special node $o^*$ before it bumps into $M$.

Hence let $A \in {\cal C}^X$ without loss of generality be the first agent which exits the {\em initial} state and bumps into $M$ when it tries to move from a node $m$, and let $B$ (which should belong to ${\cal C}^Y$ due to the FIFO links)  be the first agent which exits the {\em initial} state and bumps into $M$ when it tries to move from a node $m'$. Both agents $A$ and $B$ become {\em stoppers}. Then let $C$ be the next agent that exits the {\em initial} state and meets $A$ or $B$. Suppose without loss of generality that $C$ meets $A$ (and therefore $C \in {\cal C}^X$). Hence $C$ becomes a {\em searcher} and moves counter-clockwise. All other agents (if any) belonging in ${\cal C}^X$ will gather at node $m$. A similar situation may appear at node $m'$ if there are more agents belonging in ${\cal C}^Y$ (i.e., having the same orientation for the ring as agent $B$). Nevertheless there will be at most two {\em searchers} $C,D$ and they will have different orientation for the ring. All other agents will gather at nodes $m$ and $m'$ before the {\em searchers} $D$ and $C$ arrive at $m$ and $m'$ respectively. If the total number $k$ of the honest agents is odd, then exactly one of the sets ${\cal C}^X, {\cal C}^Y$ has an odd number of agents. Suppose without loss of generality that ${\cal C}^X$ has an odd number of agents. Then at node $m$ exactly one of the agents there, it will be at state {\em stopper} 
while at node $m'$ there will be only agents at state {\em final}, before the corresponding searcher arrives at each node. Hence only the searcher which reaches $m$ will change its state to {\em collector}, will collect all agents and gather with everyone at node $m'$ (the other searcher which arrives at $m'$ will stay there). 

In any case above, it is easy to see that each agent may traverse the network at most a constant number of times before rendezvous occurs. Hence, the total number of edge traversals (i.e., steps) for all the agents is $O(kn)$.\qed

The following result summarizes the results of this section:

\begin{theorem}\label{ring:iff}
In any anonymous and asynchronous ring with a special node \oStar and one malicious agent, $k$ honest agents having constant memory and no knowledge about their number, can solve the rendezvous problem if and only if either the ring is oriented or $k$ is odd.
\end{theorem}

We briefly consider the case when there could be multiple malicious agents in the network.
In this case, rendezvous is feasible only if all the malicious agents are located in a continuous segment of the ring with no honest agent in between. This scenario is equivalent to the one with a single malicious agent and thus the same algorithm would work in this case.



\section{Rendezvous in an Oriented Mesh Network}
\label{mesh-tori}
We now study the problem in an oriented mesh network. In view of Lemma~\ref{separable}, rendezvous is impossible for separable initial configurations. Hence, in this section we study the problem for a special class of non separable initial configurations and we give an algorithm that solves the problem for this type of configurations. In particular, we focus on initial configurations where the induced subgraph of the occupied nodes is connected without holes, i.e., there is no connected set of unoccupied nodes surrounded by occupied nodes. At the end of the section we discuss the solvability of the problem in other classes of initial non separable configurations.

First observe that even in configurations that consist of a simple path of occupied nodes, the problem is unsolvable in the considered model due to network asynchronicity: Initially all agents have the same input and thus (following any potentially correct algorithm), they should all try to move; however, an adversary may slowdown all agents, except for one not located at the endpoints of the path, hence creating a separating configuration. Thus, by Lemma~\ref{separating} the problem is unsolvable.

Therefore, the agents need to be able to gain some knowledge about their current configuration before they move in order to avoid creating a separating configuration. We enhance our model by giving the agents, the capability to discover all occupied nodes within a distance of $d$-hops. 

\begin{definition} \label{visibility} We say that an agent $A$ located at a node $x$ {\em can see (or scan) at a distance $d$} or it has {\em  d-visibility} if $A$ can decide for any node $u$ within a distance of $d$ hops from $x$, whether $u$ is occupied or not by an honest agent. 
\end{definition}

We emphasize that, if a node $u$ scanned by agent $A$ is occupied, $A$ cannot tell how many agents are in $u$, or read their states.
When the agents have a $d-$visibility capability we assume that moves are instantaneous, i.e., an agent cannot be traveling along an edge while another agent is scanning its neighbourhood. 
Unfortunately, as we show below, even when the agents have $1-$visibility (i.e., they can only scan their neighbours), the problem remains unsolvable for some connected without holes configurations.  

\begin{lemma}\label{impossible-mesh-d1}
The rendezvous problem is unsolvable in an oriented mesh with a malicious agent for initial connected without holes configurations, even when the agents are capable of scanning their adjacent nodes.
\end{lemma}

\pf
Consider any algorithm that moves an agent $A$ having exactly two adjacent occupied nodes according to its local view as shown in any of the four configurations of Figure~\ref{mesh-impossible-d1} right. 
Notice that if the algorithm also moves agents having exactly one adjacent occupied node in any of those configurations, then an adversary may execute $A$'s move and delay any other agents' moves. 
Then clearly any such algorithm cannot be correct since after $A$'s move in any configuration of Figure~\ref{mesh-impossible-d1} right, the configuration becomes separating and by Lemma~\ref{separating} the problem is unsolvable. Therefore any potentially correct algorithm should not move agents with exactly two adjacent occupied nodes and local view any of those shown in Figure~\ref{mesh-impossible-d1} right. 
However, in a similar new configuration where there is a fourth agent occupying the empty node adjacent to two occupied nodes in any configuration of Figure~\ref{mesh-impossible-d1} right, no agent can move (since any agent cannot differentiate this new configuration from one of the configurations in Figure~\ref{mesh-impossible-d1} right). Hence the agents in such a configuration cannot meet.
\begin{figure}
\vbox to 2mm{\hspace{0cm} \epsfig{file=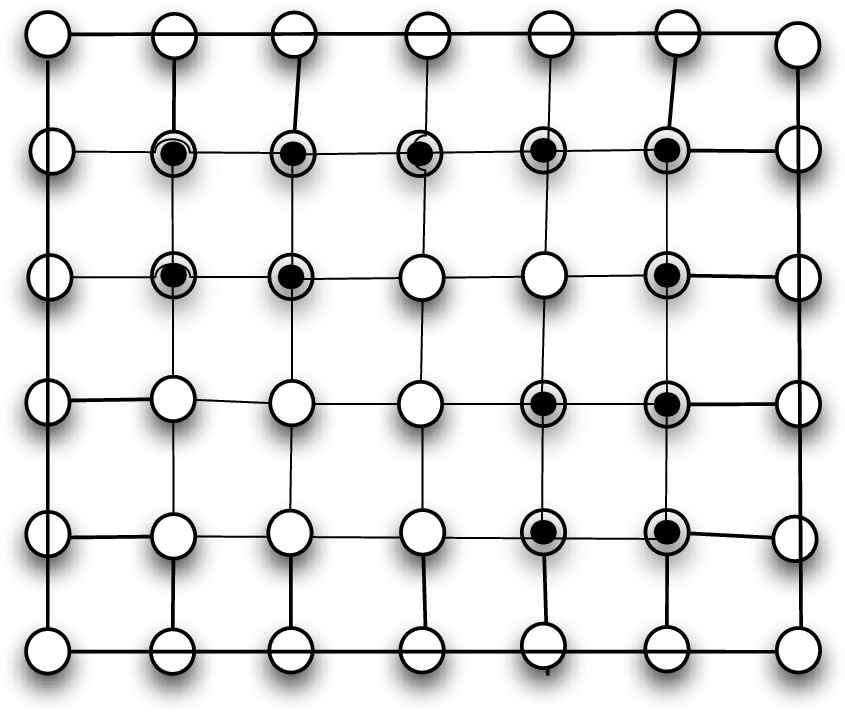, width=1in} 
\hfill \hbox to 75mm{\hspace{0cm}
\epsfig{file=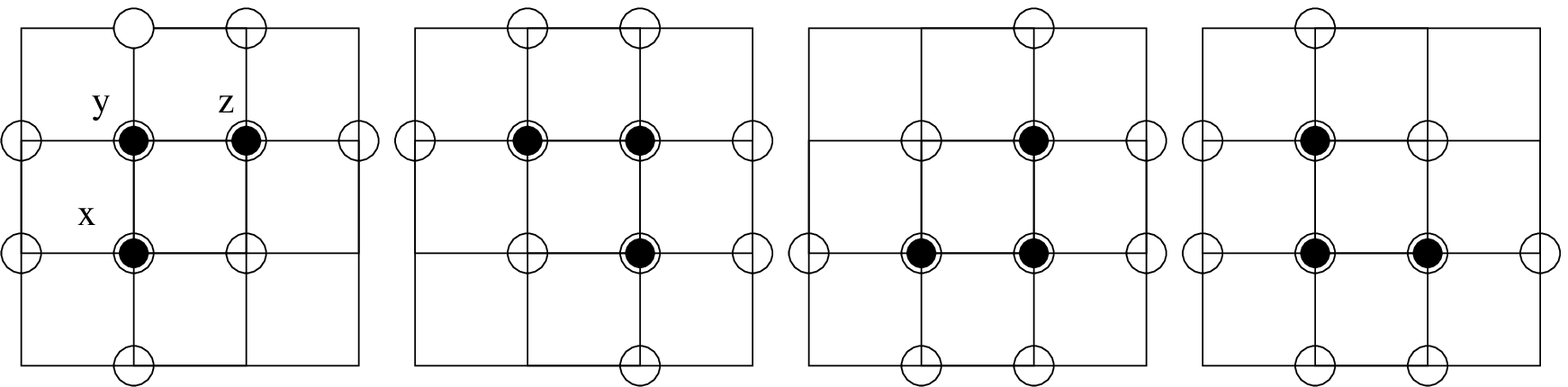, width=2.7in} }\vss}
\vspace*{2cm}
\caption{Left: Every occupied node is adjacent to at least two other occupied nodes. Right: A gamma configuration and its three rotations. \label{mesh-impossible-d1}}
\end{figure}
\qed

Additionally, as Figure~\ref{mesh-impossible-d1} left demonstrates, an agent either having: 1) exactly two adjacent occupied nodes on their West and East or, 2) exactly two adjacent occupied nodes on their South and North or, 3) three adjacent occupied nodes, cannot move since  the resulting configuration will be separating. Therefore, for any algorithm that moves an agent adjacent to two or three occupied nodes there are infinitely many configurations for which this algorithm fails to gather the agents (i.e., configurations where every occupied node is either adjacent to two or to three occupied nodes).


Hence we further equip the agents with the capability of discovering the occupied nodes within a two-hops distance. In that case, as we show below, the problem can be solved for any connected without holes initial configuration. 

We present an algorithm which instructs the agents to move only to occupied nodes in a way that they maintain the connectivity and they do not create holes. In order to describe the algorithm we define eleven local configurations as shown in Fig.~\ref{mesh11}. In these configurations, empty circles represent free nodes, while circles containing black dots represent occupied nodes. The remaining vertices on the figures represent nodes which may be either occupied or free. The agent (let us call it $A$) which is located below a horizontal arrow in cases (a-g), moves horizontally as depicted by the arrow. The agent (let us call it $B$) which is located left of a vertical arrow in cases (h-m), moves vertically as depicted by the arrow. Hence the algorithm can be described as follows: 

\smallskip
\noindent {\bf Algorithm {\sc RV-Mesh}}: {\em If an agent has a view (within two hops) like the one of agent $A$ or  $B$ described before, then this agent moves towards the direction shown by the corresponding arrow; otherwise the agent does not move.} 

\smallskip
Nodes which are within two hops from the scanning agent and are not shown in those configurations can be either occupied or free.
If the location of the scanning agent is close to the border of the mesh and some of the nodes in those eleven configurations do not exist, then the agent acts as it would act if those nodes existed in its view and were free. Moreover, while an agent $A$ located at a node $u$ is executing its scan or compute phase then no other operation can take place at $u$ before $A$ moves or decides to stay (i.e., no other agent at $u$ can start scanning and no other agent can arrive at $u$). That is, operations at a node $u$ are executed in mutual exclusion. Notice that if two adjacent agents want to swap positions they can only do it at the same time. 
We now present three lemmas.

 \begin{figure}
 \centering 
 \epsfig{file=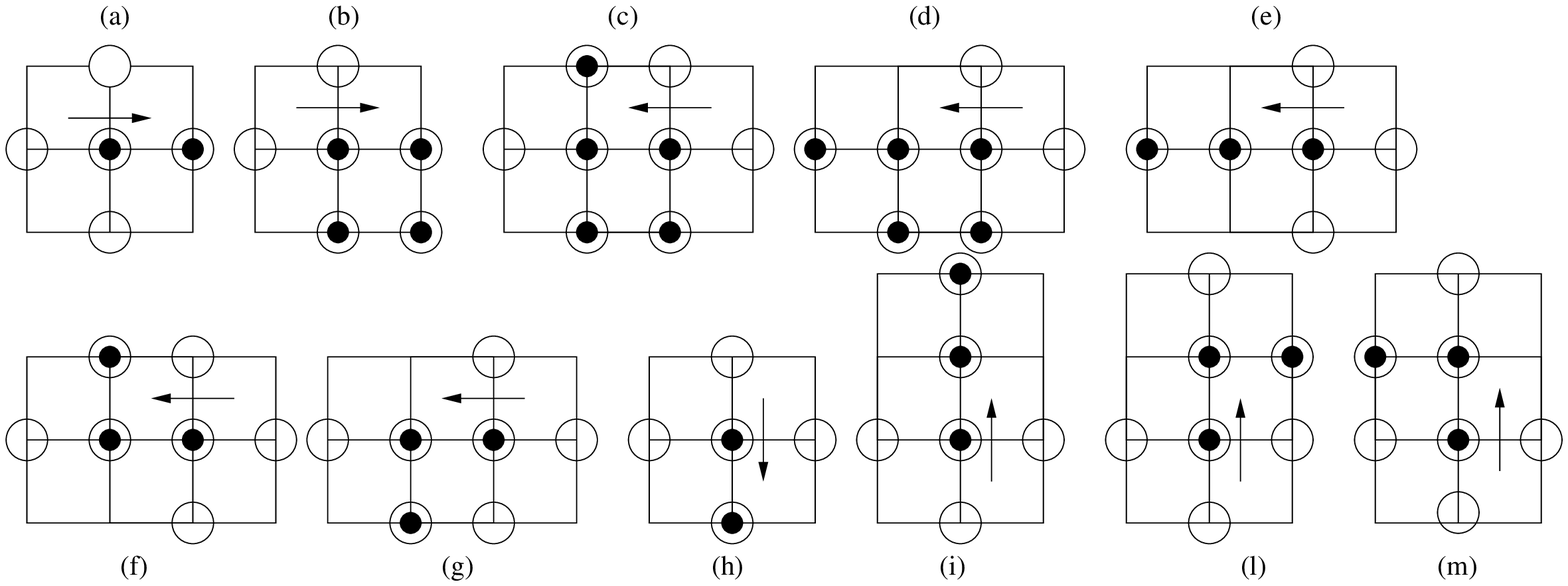, width=8.3cm}
 \caption{View of the scanning agent located below (cases a-g) or left (cases h-m) of the depicted arrow. 
 Occupied nodes are depicted as cycles containing black dots, while free nodes are depicted as empty cycles. 
 Nodes which are within two hops from the scanning agent but not shown, can be either occupied or free. The scanning agent  will move East in cases $(a, b)$, West in cases $(c, d, e, f, g)$, South in case $(h)$, and North in cases $(i, l, m)$. \label{mesh11}} 
\end{figure}

\begin{lemma}\label{mesh-correct-start}
Given an $n\times m$ oriented mesh, for any connected configuration without holes of at least three occupied nodes, there is at least one agent whose view is in one of the configurations depicted in Fig.~\ref{mesh11}.
\end{lemma}

\pf
Select the column with at least one occupied node which is closer to the left (West) border of the mesh. From that column select the upper (North-West) 
occupied node $u$. Hence the two nodes West  and North of $u$ are free. Since the configuration is connected, at least one of the nodes East (say $w$) or South (say $v$) of $u$, is occupied. If only $w$ is occupied then the agents at $u$ have the view of Fig.~\ref{mesh11} (a) and move towards East. If only $v$ is occupied then the agents at $u$ have the view of Fig.~\ref{mesh11} (h) and move towards South. If both $w, v$ are occupied then if the node which is East of $v$ and South of $w$ (say $z$) is occupied, the agents at $u$ have the view of Fig.~\ref{mesh11} (b) and move towards East. If however, node $z$ is free, then the agents at $u$ do not move. In this last case, since the configuration is connected and without  holes, and thus cycles, it is composed of the following two `generalised' trees: One tree starts with edge $(u,v)$ and the other starts with edge $(u,w)$. Each occupied node of the configuration different from $u$, appears in exactly one of those two trees (otherwise there would be a cycle of occupied nodes containing node $z$).
We call these trees `generalised' because some parts of the trees could be `thick' (i.e., composed of more than one row or column of occupied nodes). Notice that there should be a frontier of free nodes which extends from node $z$ to either the East or the South border of the mesh, which prevents any two occupied nodes belonging to different trees, from being connected through any path of occupied nodes that is not passing from node $u$ (since if such a path exists for two nodes, then the free node $z$ would be contained in a cycle of occupied nodes). We can now show that at least one agent on some branch of a tree will move according to one of the views of Fig.~\ref{mesh11}. In fact we prove that the branches of the tree (which extend East, West, North or South) shrink starting from their ends. First it is easy to see that for occupied nodes which are connected only to their East or to their South, the agents there, have views of Fig.~\ref{mesh11} (a) and (h) respectively, and they will move accordingly. For occupied nodes which are connected only to their West or to their North, the agents there, have views of Fig.~\ref{mesh11} (e,f,g) and  (i,l,m) respectively, and they will move accordingly. For branches with blocks of at least $4$ occupied nodes:
\begin{enumerate}
\item If the branch extends to the West then its upper leftmost occupied node $x$ is connected only to its East and South, the agents at $x$ have the view of Fig.~\ref{mesh11} (b) and they move towards East. After repeating such moves and some others based on the view of Fig.~\ref{mesh11} (a), the branch will be eliminated.
\item If the branch extends to the North then its upper leftmost occupied node $x$ is connected only to its East and South, the agents at $x$ have the view of Fig.~\ref{mesh11} (b) and they move towards East. After repeating such moves and some others based on the view of Fig.~\ref{mesh11} (h) and maybe some moves based on the views of Fig.~\ref{mesh11} (c,d), the branch will be eliminated. 
\item If the branch extends to the East then its upper rightmost occupied node $y$ is connected only to its West and South, the agents at $y$ have the view of Fig.~\ref{mesh11} (c,d) and they move towards West. After repeating such moves and some others based on the view of Fig.~\ref{mesh11} (e,f,g), the branch will be eliminated.
\item If the branch extends to the South then:
	\begin{itemize} 
	\item if its upper leftmost occupied node $x$ is connected only to its East and South or its upper rightmost occupied node $y$ is connected only to its West and South, the agents at $x$ or $y$ can repeatedly move based on the views of Fig.~\ref{mesh11} (b,a,i,l), or (c,d,e,f,g) respectively, and the branch will be eliminated.
	\item if both its upper leftmost occupied node $x$ is connected to its West or North
	and its upper rightmost occupied node $y$ is connected to its East and North  
	 then in at least one of those nodes $x,y$ those extra branches can be eliminated as above so that case $3$ above will hold.
	\end{itemize}
\end{enumerate}
\qed

\begin{lemma}\label{mesh-correct-meet}
Given an $n\times m$ oriented mesh, consider a connected configuration of $k$ agents in two occupied nodes. According to Algorithm {\sc RV-Mesh}, after a total number of at most $k+1$ edge traversals there will be only one occupied node.
\end{lemma}

\pf
If the occupied nodes are in a row, then it may happen that two agents move based on views (not perceived from the same present configuration) of Fig.~\ref{mesh11} (a and e), swapping places. However, then only the agents on the West will move according to view of the Fig.~\ref{mesh11} (a), and they all gather at one node. 

If the occupied nodes are in a column, then it may happen that two agents move based on views (not perceived from the same configuration) of Fig.~\ref{mesh11} (h and i), swapping places. However, then only the agents on the North will move according to view of the Fig.~\ref{mesh11} (h), and they all gather at one node.  
\qed

\begin{lemma}\label{mesh-correct-decrease}
Given an $n\times m$ oriented mesh, consider any connected configuration without holes of $k$ agents occupying at least $3$ nodes. After any number of moves according to Algorithm {\sc RV-Mesh}, the resulting configuration is also connected without holes. Furthermore, the number of occupied nodes will strictly decrease after  at most $k$ edge traversals, reaching the value of only one occupied node after at most $O(k^2)$ edge traversals.
\end{lemma}

\pf
Let ${\cal C}$ be any connected configuration without holes of $k$ agents occupying at least three nodes and let ${\cal C}'$ be a resulting configuration after any number of agents' moves according to Algorithm {\sc RV-Mesh}. 
First observe that since any agent can only move to an adjacent node which was occupied at the moment the agent scanned, this means that no agent will ever move to a node which was not occupied at ${\cal C}$. Hence ${\cal C}'$ consists of a subset of occupied nodes of ${\cal C}$, and since an agent which has all its neighbours occupied never moves, if ${\cal C}'$ contains holes then it would mean that ${\cal C}$ also contained holes, which is a contradiction. 

Now we show that the new configuration ${\cal C}'$ remains connected. Notice that an agent $A$ moves towards North or South only when it is adjacent to only one other 
occupied node $v$ and moves towards $v$ (see cases (h,i,l,m) of Fig.~\ref{mesh11}). Furthermore, the agents at node $v$ cannot move according to any view before $A$ moves. Hence if $A$ moves towards North or South it cannot disconnect the configuration. If an agent $A$ moves towards West (according to cases (e,f,g) of Fig.~\ref{mesh11}) or East (according to case (a) of Fig.~\ref{mesh11}) then again it cannot disconnect the configuration since $A$ has to move towards its only neighbouring occupied node and agents from that node cannot move according to any view before $A$ moves. For the remaining cases ((b,c,d) of Fig.~\ref{mesh11}) we have: i) When an agent $A$ occupies the upper leftmost position of a block of at least $4$ nodes (as shown in case (b)) it moves towards East and only the upper rightmost agent $B$ of the block could be willing to move at the same time towards West based on an earlier view (according to cases (c,d)). Then $A,B$ can move at the same time. It is easy to see that the lower two agents of the block cannot move by any view. ii) Similarly if an agent occupies the upper rightmost position of a block of at least $4$ nodes (as shown in cases (c,d)) it moves towards West and only the upper leftmost agent of the block could be willing to move at the same time towards East based on an earlier view (according to case (a)). Hence the configuration remains connected.

Finally, we show that a configuration ${\cal C}'$ that has resulted from ${\cal C}$ after at most $k$ edge traversals has strictly less occupied nodes than  ${\cal C}$. As proved before, either all co-located agents will move (one by one) according to one of the views of Fig.~\ref{mesh11} making a new free node which will never be reoccupied or at a block of at least $4$ occupied nodes the top agents will swap positions. However, according to the views of Fig.~\ref{mesh11}, this swapping may happen only once for a fixed block (since the views change). Hence after a total number of at most $(k-1)/2$ swaps and one additional move the occupied nodes in ${\cal C}'$ will be strictly less than in ${\cal C}$.
Thus after at most $O(k^2)$ edge traversals the agents will rendezvous.
\qed

\noindent In view of Lemmas~\ref{impossible-mesh-d1}, \ref{mesh-correct-start}, \ref{mesh-correct-meet}, \ref{mesh-correct-decrease} we have:
\begin{theorem}\label{mesh-correct}
The rendezvous problem for $k\geq 2$ agents can be solved 
for any initial connected without holes configuration of agents in an $n\times m$ oriented mesh if and only if the agents are able to discover the occupied nodes within a distance of two-hops.  
\end{theorem}

If the initial non separable configuration is different from the one considered above, then even the $2$-visibility capability is not sufficient anymore to solve rendezvous. In fact Fig.~\ref{mesh-holes}, demonstrates that the problem remains unsolvable for connected configurations with holes 
even when the agents are able to discover the occupied nodes within any constant distance. 
\begin{figure}
 \centering \epsfig{file=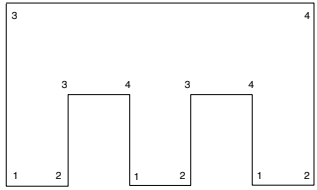, width=4.3cm}
 \caption{An initially connected configuration of agents with a hole. Even if the agents are able to discover the occupied nodes within any constant distance, this `polygonal' configuration can always be expanded so that all agents with the same number have exactly the same view and therefore at least two of them should move. However, after those moves the resulting configuration is separating and the problem is unsolvable.\label{mesh-holes}}
 \end{figure}
The problem is also unsolvable for some disconnected non separable configurations (see Fig.~\ref{mesh-discon}). 
\begin{figure}
 \centering \epsfig{file=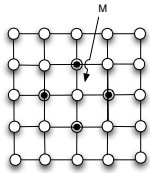, width=2.3cm}
 \caption{A disconnected non separable initial configuration of agents. The malicious agent has been placed as depicted. If any agent moves then the resulting configuration is separating and the problem is unsolvable.\label{mesh-discon}}
 \end{figure}
Hence it appears that for many initial non separable configurations in an oriented mesh, the combination of the asynchronicity and the limited view (to any constant fraction of the complete view) makes the problem unsolvable.


\section{Conclusion}
\label{conclusion}
In this paper we studied deterministic protocols for the rendezvous of $k\geq 2$ honest agents in asynchronous networks with a malicious agent which can prevent the agents from reaching any node it occupies. We have presented algorithms for oriented and unoriented ring networks  which gathers the honest agents within $O(kn)$ edge traversals for all feasible instances of the problem. We have also presented a deterministic protocol for oriented $n \times m$ meshes which leads the agents to rendezvous within $O(k^2)$ edge traversals for any initial connected without holes configuration if and only if the agents can discover the occupied nodes within a distance of two-hops.
 Given the novelty of the model there are many interesting open questions. 
The first is whether the problem can be solved in unoriented meshes for connected configurations without holes when the agents are capable of scanning within a constant distance (more than one). 
 It would be also interesting to study randomized protocols for some of the unsolvable cases, and also to study this problem in synchronous networks with unit-speed cooperating agents and unit-speed/infinite-speed malicious agents. 
Finally,  it would be interesting to study the problem in $(m+1)$-connected graphs in the presence of $m$ malicious agents, or in the solved cases presented in this paper in the presence of malicious agents that show a  more severe behaviour.

\newpage
\pagenumbering{roman}

\bibliographystyle{abbrv}
\bibliography{biblio}

\begin{thebibliography}{10}

\bibitem{ap06}
N.~Agmon and D.~Peleg.
\newblock Fault-tolerant gathering algorithms for autonomous mobile robots.
\newblock {\em SIAM J. on Computing}, 36(1):56--82, 2006.

\bibitem{ag02}
S.~Alpern and S.~Gal.
\newblock Searching for an agent who may or may not want to be found.
\newblock {\em Operations Research}, 50(2):311--323, 2002.

\bibitem{blmpp14}
E.~Bampas, N.~Leonardos, E.~Markou, A.~Pagourtzis, and M.~Petrolia.
\newblock Improved periodic data retrieval in asynchronous rings with a faulty
  host.
\newblock In {\em SIROCCO}, LNCS 8576, pages 355--370, 2014.

\bibitem{bfffnst12}
L.~Barriere, P.~Flocchini, F.~V. Fomin, P.~Fraigniaud, N.~Nisse, N.~Santoro,
  and D.~Thilikos.
\newblock Connected graph searching.
\newblock {\em Information and Computation}, 219:1--16, 2012.

\bibitem{Bouzid0T13}
Z.~Bouzid, S.~Das, and S.~Tixeuil.
\newblock Gathering of mobile robots tolerating multiple crash faults.
\newblock In {\em {IEEE} 33rd International Conference on Distributed Computing
  Systems, {ICDCS} 2013, 8-11 July, 2013, Philadelphia, Pennsylvania, {USA}},
  pages 337--346, 2013.

\bibitem{cd10}
J.~Chalopin and S.~Das.
\newblock Rendezvous of mobile agents without agreement on local orientation.
\newblock In {\em ICALP}, LNCS 6199, pages 515--526, 2010.

\bibitem{cds07}
J.~Chalopin, S.~Das, and N.~Santoro.
\newblock Rendezvous of mobile agents in unknown graphs with faulty links.
\newblock In {\em Proc. of 21st Int. Conf. on Distributed Computing}, pages
  108--122, 2007.

\bibitem{cdlp14}
J.~Chalopin, Y.~Dieudonne, A.~Labourel, and A.~Pelc.
\newblock Fault-tolerant rendezvous in networks.
\newblock In {\em Proc. of 41st Int. Colloquium on Automata, Languages and
  Programming}, LNCS 8573, pages 411--422, 2014.

\bibitem{cp08}
R.~Cohen and D.~Peleg.
\newblock Convergence of autonomous mobile robots with inaccurate sensors and
  movements.
\newblock {\em SIAM J. on Computing}, 38:276--302, 2008.

\bibitem{clp10}
J.~Czyzowicz, A.~Labourel, and A.~Pelc.
\newblock How to meet asynchronously (almost) everywhere.
\newblock In {\em Proc. of 21st Annual ACM-SIAM Symp. on Discrete Algorithms},
  2010.

\bibitem{dmsvw08}
S.~Das, M.~Mihalak, R.~Sramek, E.~Vicari, and P.~Widmayer.
\newblock Rendezvous of mobile agents when tokens fail anytime.
\newblock In {\em OPODIS}, LNCS 5401, pages 463--480, 2008.

\bibitem{dpp14}
Y.~Dieudonne, A.~Pelc, and D.~Peleg.
\newblock Gathering despite mischief.
\newblock {\em ACM Transactions on Algorithms}, 11(1):1, 2014.

\bibitem{dfps03}
S.~Dobrev, P.~Flocchini, G.~Prencipe, and N.~Santoro.
\newblock Multiple agents rendezvous in a ring in spite of a black hole.
\newblock In {\em OPODIS}, pages 34--46, 2003.

\bibitem{DBLP:journals/algorithmica/DobrevFPS07}
S.~Dobrev, P.~Flocchini, G.~Prencipe, and N.~Santoro.
\newblock Mobile search for a black hole in an anonymous ring.
\newblock {\em Algorithmica}, 48(1):67--90, 2007.

\bibitem{fhl07}
P.~Flocchini, M.~J. Huang, and F.~L. Luccio.
\newblock Decontamination of chordal rings and tori using mobile agents.
\newblock {\em Int. Jour. of Foundation of Comp. Sc.}, 3(18):547--564, 2007.

\bibitem{fhl08}
P.~Flocchini, M.~J. Huang, and F.~L. Luccio.
\newblock Decontamination of hypercubes by mobile agents.
\newblock {\em Networks}, 3(52):167--178, 2008.

\bibitem{DBLP:journals/algorithmica/FlocchiniIS12}
P.~Flocchini, D.~Ilcinkas, and N.~Santoro.
\newblock Ping pong in dangerous graphs: Optimal black hole search with
  pebbles.
\newblock {\em Algorithmica}, 62(3-4):1006--1033, 2012.

\bibitem{:FlocchiniS2012}
P.~Flocchini and N.~Santoro.
\newblock Distributed security algorithms for mobile agents.
\newblock In J.~Cao and S.~K. Das, editors, {\em Mobile Agents in Networking
  and Distributed Computing}, chapter~3, pages 41--70. John Wiley \& Sons,
  Inc., Hoboken, NJ, USA, 2012.

\bibitem{kmrs07}
R.~Klasing, E.~Markou, T.~Radzik, and F.~Sarracco.
\newblock Hardness and approximation results for black hole search in arbitrary
  graphs.
\newblock {\em TCS}, 384(2-3):201--221, 2007.

\bibitem{DBLP:conf/sirocco/KralovicM10}
R.~Kr{\'a}lovi{\v c} and S.~Mikl{\'\i}k.
\newblock Periodic data retrieval problem in rings containing a malicious host.
\newblock In {\em SIROCCO}, LNCS 6058, pages 157--167, 2010.

\bibitem{kkmbook}
E.~Kranakis, D.~Krizanc, and E.~Markou.
\newblock {\em The Mobile Agent Rendezvous Problem in the Ring}.
\newblock Synthesis Lectures on Distributed Computing Theory. Morgan \&
  Claypool Publishers, 2010.

\bibitem{l09}
F.~L. Luccio.
\newblock Contiguous search problem in sierpinski graphs.
\newblock {\em Theory of Comp. Sys.}, (44):186--204, 2009.

\bibitem{YIK12}
Y.~Yamauchi, T.~Izumi, and S.~Kamei.
\newblock Mobile agent rendezvous on a probabilistic edge evolving ring.
\newblock In {\em ICNC}, pages 103--112, 2012.

\bibitem{yy96}
X.~Yu and M.~Yung.
\newblock Agent rendezvous: A dynamic symmetry-breaking problem.
\newblock In {\em ICALP}, pages 610--621, 1996.

\bibitem{YXS10}
Y.~Zeng, X.~Hu, and K.~Shin.
\newblock Detection of botnets using combined host- and network-level
  information.
\newblock In {\em IEEE/IFIP DSN 2010}, pages 291--300, June 2010.

\end{thebibliography}

\end{document}